\documentclass[runningheads]{llncs}

\usepackage{graphicx}
\usepackage{wrapfig}
\usepackage{amssymb}
\usepackage{amsmath}
\usepackage{latexsym}
\usepackage{graphicx}
\usepackage{booktabs}
\usepackage{float}
\usepackage{comment}
\usepackage[shortlabels]{enumitem} 
\usepackage[utf8]{inputenc} 
\usepackage{listings}
\usepackage[dvipsnames]{xcolor}
\usepackage[noend]{algpseudocode}
\usepackage{algorithm,algorithmicx}
\usepackage{url}
\usepackage[font=footnotesize]{caption} 
\usepackage[misc]{ifsym} 

\newcommand{\Gen}{\textsc{Gen}}
\newcommand{\RevGen}{\textsc{RevGen}}

\newcommand{\LIST}[1]{\mathcal{G}_{#1}}
\newcommand{\REVLIST}[1]{\mathcal{R}_{#1}}
\newcommand{\trees}[1]{\mathbf{T}_{#1}}
\newcommand{\spt}[1]{\ensuremath{t_{#1}}}

\algrenewcommand\algorithmicrequire{\textbf{Precondition:}}
\algrenewcommand\algorithmicensure{\textbf{Postcondition:}}

\definecolor{verbgray}{gray}{0.9}
\lstnewenvironment{code}{%
  \lstset{backgroundcolor=\color{verbgray},
 frame=single,
  language=C,
  framerule=0pt,
  basicstyle=\ttfamily,
  showstringspaces=false,
  commentstyle=\color{blue}\textit,
  keywordstyle=\color{black}\bf,
  numbers=none,
  keepspaces=true,
  columns=fullflexible}}{}

\begin{document}
\title{A Pivot Gray Code Listing for the Spanning Trees of the Fan Graph}

\author{ Ben Cameron
\and
Aaron Grubb $^{\textrm{\Letter}}$
\and
Joe Sawada
}

\authorrunning{B. Cameron et al.}
%
\institute{University of Guelph, Guelph, Canada 
\email{\{ben.cameron,agrubb,jsawada\}@uoguelph.ca}
}
\maketitle   

\begin{abstract}
\vspace{-0.2in}
We use a greedy strategy 
to list the spanning trees of the fan graph, $F_n$, such that successive trees differ by pivoting a single edge around a vertex. It is the first greedy algorithm for exhaustively generating spanning trees using such a minimal change operation. The resulting listing is then studied to find a recursive algorithm that produces the same listing in $O(1)$-amortized time using $O(n)$ space.  Additionally, we present $O(n)$-time algorithms for ranking and unranking the spanning trees for our listing; an improvement over the generic $O(n^3)$-time algorithm for ranking and unranking spanning trees of an arbitrary graph.


\end{abstract}

\section{Introduction} \label{sec:intro}

This paper is concerned with the algorithmic problem of listing all spanning trees of the fan graph. Applications of efficiently listing all spanning trees of general graphs are ubiquitous in computer science and also appear in many other scientific disciplines \cite{chakraborty}. In fact, one of the earliest known works on listing all spanning trees of a graph is due to the German physicist Wilhelm Feussner in 1902 who was motivated by an application to electrical networks \cite{1902}. In the 120 years since Feussner's work, many new algorithms have been developed, such as those in the following citations \cite{berger,Char,cummins,gabow,hakimi,holzmann,kamae,kapoor,kishi,matsui,mayeda,minty,Shioura1995,uno,smith1997generating,winter}.

For any application, it is desirable for spanning tree listing algorithms to have the asymptotically best possible running time, that is, $O(1)$-amortized running time. The algorithms due to Kapoor and Ramesh \cite{kapoor}, Matsui \cite{matsui}, Smith \cite{smith1997generating}, Shioura and Tamura \cite{Shioura1995} and Shioura et al. \cite{uno}  all run in $O(1)$-amortized time. Another desirable property of such listings is to have the \emph{revolving-door} property, where successive spanning trees differ by the addition of one edge and the removal of another. Such listings where successive objects in a listing differ by a constant number of simple operations are more generally known as \textit{Gray codes}. The algorithms due to Smith \cite{smith1997generating}, Kamae \cite{kamae}, Kishi and Kajitani \cite{kishi}, Holzmann and Harary \cite{holzmann} and Cummins \cite{cummins} all produce Gray code listings of spanning trees for an arbitrary graph.  Of all of these algorithms, Smith's is the only one that produces a Gray code listing in $O(1)$-amortized time.  A stronger notion of a Gray code for spanning trees is where the revolving-door makes strictly local changes.  More specifically, we would like the differing edges to share a common endpoint.  Such a Gray code property, which we call a \textit{pivot Gray code}, is not given by any previously known algorithm. This leads to our first research question.

\begin{quote} \small
{\bf Research Question \#1} Given a graph $G$ (perhaps from a specific class), does there exist a pivot Gray code listing of all spanning trees of $G$? Furthermore, can the listing be generated in polynomial (ideally constant) time per tree using polynomial space?
\end{quote}

A related question that arises for any listing is how to \emph{rank}, that is, find the position of the object in the listing,  and \emph{unrank}, that is, return the object at a specific rank. For spanning trees, an $O(n^3)$-time algorithm for ranking and unranking a spanning tree of a specific listing for an arbitrary graph is known~\cite{colbourn1989unranking}. 

\begin{quote} \small
{\bf Research Question \#2} Given a graph $G$ (perhaps from a specific class), does there exist a (pivot Gray code) listing of all spanning trees of $G$ that can be ranked and unranked in $O(n^2)$ time or better? 
\end{quote}

An algorithmic technique recently found to have success in the discovery of Gray codes is the greedy approach.  An algorithm is said to be \emph{greedy} if it can prioritize allowable actions according to some criteria, and then choose the highest priority action that results in a unique object to obtain the next object in the listing.
When applying a greedy algorithm, there is no backtracking; once none of the valid actions lead to a new object in the set under consideration, the algorithm halts, even if the listing is not exhaustive.  The work by Williams~\cite{williams2013greedy} notes that some very well-known combinatorial listings can be constructed greedily, including the binary reflected Gray code (BRGC) for binary strings, the plain change order for permutations,  and the lexicographically smallest de Bruijn sequence. 
Recently, a very powerful greedy algorithm on permutations (known as Algorithm J, where J stands for ``jump'') generalizes many known combinatorial Gray code listings including many related to permutation patterns, rectangulations, and elimination trees~\cite{MUTZE2020,MUTZEHoang2019,MUTZERectangulations2021}. However, no greedy algorithm was previously known to list the spanning trees of an arbitrary graph.

\begin{quote}  \small
{\bf Research Question \#3} Given a graph $G$ (perhaps from a specific class), does there exist a greedy strategy to list all spanning trees of $G$?  Moreover, does such a greedy strategy exist where the resulting listing is a pivot Gray code? 
\end{quote}

\noindent
In most cases, a greedy algorithm requires exponential space to recall which objects have already been visited in a listing. Thus, answering this third question would satisfy only the first part of {\bf Research Question \#1}. However, in many cases, an underlying pattern can be found in a greedy listing which can result in space efficient algorithms~\cite{MUTZE2020,williams2013greedy}. 

To address these three research questions, we applied a variety of greedy approaches to structured classes of graphs including the fan, wheel, $n$-cube, and the compete graph.  From this study, we were able to affirmatively answer each of the research questions for the fan graph.  It remains an open question to find similar results for other classes of graphs.

\subsection{New Results} \label{sec:results}

The \textit{fan graph} on $n$ vertices, denoted $F_n$, is obtained by joining a single vertex (which we label $v_\infty$) to the path on $n-1$ vertices (labeled $v_2, ... , v_n$) -- see Fig.~\ref{fig:F5}.   
\begin{wrapfigure}[7]{r}{0.35\textwidth}
\begin{center}
  \vspace*{-0.5cm}
  \hspace{-0.6cm}
  \includegraphics[scale=0.35, trim=0 0 0 0cm, clip]{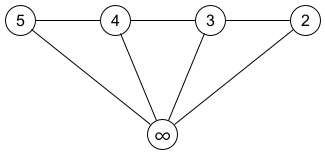}
  \caption{$F_5$}
  \vspace{-0.2cm}
  \label{fig:F5}
\end{center}
\end{wrapfigure}
Note that we label the smallest vertex $v_2$ so that the largest non-infinity labeled vertex equals the total number of vertices. Let $\trees{n}$ denote the set of all spanning trees of $F_n$.
We discover a greedy strategy to generate $\trees{n}$ in a pivot Gray code order. We describe this greedy strategy in Section 2. The resulting listing is studied to find an $O(1)$-amortized time recursive algorithm that produces the same listing using only $O(n)$ space, which is presented in Section 3. We also show how to rank and unrank a spanning tree of the greedy listing in $O(n)$ time in Section 3, which is a significant improvement over the general $O(n^3)$-time ranking and unranking that is already known. We conclude with a summary in Section 4. 


\section{A Greedy Generation for $\trees{n}$}

With our goal to discover a pivot Gray code listing of $\trees{n}$, we tested a variety of greedy approaches.  
There are two important issues when considering a greedy approach to list spanning trees: (1) the labels on the vertices (or edges) and (2) the starting tree. For each of our approaches, we prioritized our operations by first considering which vertex $u$ to pivot on, followed by an ordering of the endpoints considered in the addition/removal.  We call the vertex $u$ the \emph{pivot}.  

Our initial attempts focused only on pivots that were leaves.   As a specific example, we ordered the leaves (pivots) from smallest to largest.  Since each leaf $u$ is attached to a unique vertex $v$ in the current spanning tree, we then considered the neighbours $w$ of $u$ in increasing order of label.
We restricted the labeling of the vertices to the most natural ones, such as the one presented in Section~\ref{sec:results}.  For each strategy we tried all possible starting trees. Unfortunately, none of our attempts lead to exhaustive listings.  Applying these strategies on the wheel, $n$-cube, and complete graph was also unsuccessful.

By allowing the pivot to be any arbitrary vertex, we experimentally discovered several exhaustive listings for $\trees{n}$ for $n$ up to 12 (testing every starting tree for $n=12$ took about eight hours).  One listing stood out as having an easily defined starting tree as well as a nice pattern which we could study to construct the listing more efficiently.  It applied the labeling of the vertices as described in  Section~\ref{sec:results} with the following prioritization of pivots and their incident edges: 
\begin{quote}  
Prioritize the pivots $u$ from smallest to largest and then for each pivot, prioritize the edges $uv$ that can be removed from the current tree in increasing order of the label on $v$, and for each such $v$, 
prioritize the edges $uw$ that can be added to the current tree in increasing order of the label on $w$.
\end{quote}
Since this is a greedy strategy, if an edge pivot results in a spanning tree that has already been generated or a graph that is not a spanning tree, then the next highest priority edge pivot is attempted. Let \textsc{Greedy}$(T)$ denote the listing that results from applying this greedy approach starting with the spanning tree $T$. The starting tree that produced a nice exhaustive listing was the path $v_\infty, v_2, v_3, \ldots, v_n$, denoted $P_n$ throughout the paper. 
Fig.~\ref{fig:F2F3F4F5} shows the listings  \textsc{Greedy}$(P_n)$ for $n=2,3,4,5$.  The listing  \textsc{Greedy}$(P_6)$ is illustrated in Fig.~\ref{fig:F6Generation}.  It is worth noting that starting with the path $v_\infty,v_n, v_{n-1}, \ldots, v_2$ or the star (all edges incident to $v_\infty$) did not lead to an exhaustive listing of $\trees{n}$. 

As an example of how the greedy algorithm proceeds, consider the listing  \textsc{Greedy}$(P_5)$ in Fig.~\ref{fig:F2F3F4F5}.  When the current tree $T$ is the 16th one in the listing (the one with edges $\{v_2v_\infty, v_2v_3, v_3v_4, v_5v_\infty\}$), the first pivot considered is $v_2$.  Since both $v_2v_3$ and $v_2v_\infty$ are present in the tree, no valid move is available by pivoting on $v_2$.  The next pivot considered is $v_3$.  Both edges $v_3v_2$ and $v_3v_4$ are incident with $v_3$. First, we attempt to remove $v_3v_2$ and add $v_3v_\infty$, which results in a tree previously generated. Next, we attempt to remove $v_3v_4$ and add $v_3v_\infty$, which results in a cycle. So, the next pivot, $v_4$, is considered.  The only edge incident to $v_4$ is $v_4v_3$.  By removing $v_4v_3$ and adding $v_4v_5$ we obtain a new spanning tree, the next tree in the greedy listing.

 To prove that \textsc{Greedy}$(P_n)$ does in fact contain all trees in $\trees{n}$, we
demonstrate it is equivalent to a recursively constructed listing that we obtain by studying the greedy listings.  Before we describe this recursive construction we mention one rather remarkable property of \textsc{Greedy}$(P_n)$ that we will also prove in the next section:  If $X_n$ is last tree in the listing \textsc{Greedy}$(P_n)$, then \textsc{Greedy}$(X_n)$ is precisely \textsc{Greedy}$(P_n)$ in reverse order.

\begin{figure}
\begin{center}
  \includegraphics[scale=0.25]{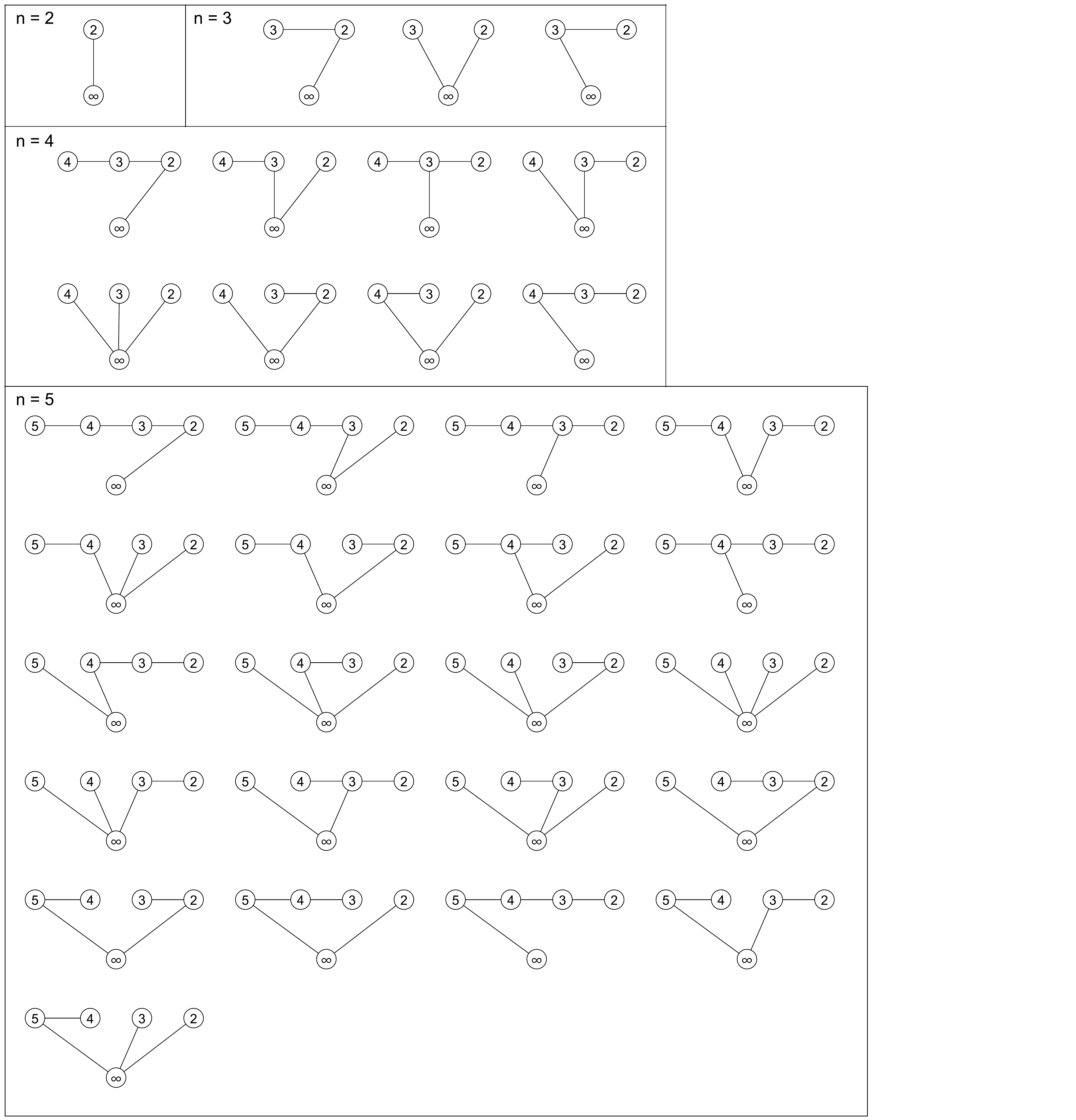}
  \caption{\textsc{Greedy}$(P_n)$ for $n=2,3,4,5$. Read left to right, top to bottom.}
  \label{fig:F2F3F4F5}
\end{center}
\end{figure}

\begin{figure}
\begin{center}
  \includegraphics[scale = 0.28]{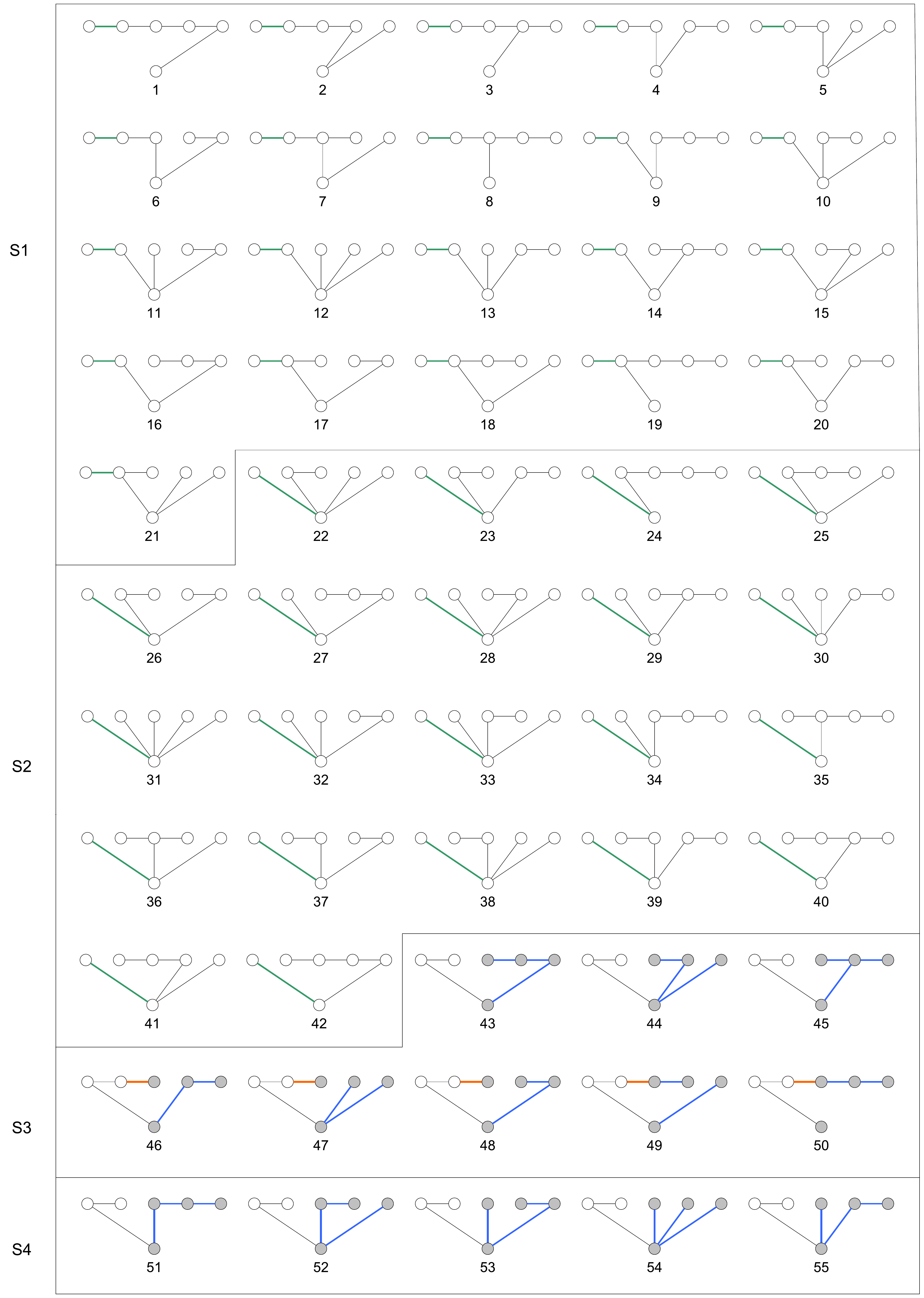}
  \caption{\textsc{Greedy}$(P_6)$ read from left to right, top to bottom. Observe that S1 is \textsc{Greedy}$(P_5)$ with $v_6 v_5$ added,  S2 is the reverse of \textsc{Greedy}$(P_5)$ with $v_6 v_\infty$ added,  S3 is \textsc{Greedy}$(P_4)$ with $v_6 v_5$ and $v_6 v_\infty$ added, except the edge $v_4 v_\infty$ is replaced by $v_4 v_5$, and S4 is the last five trees of \textsc{Greedy}$(P_4)$ in reverse order ($v_4 v_\infty$ is now present) with $v_6 v_5$ and $v_6 v_\infty$ added.}
  \label{fig:F6Generation}
\end{center}
\end{figure}


\section{An $O(1)$-amortized time Pivot Gray Code Generation for $\trees{n}$}

In this section we develop an efficient recursive algorithm to construct the listing  \textsc{Greedy}$(P_n)$.  The construction generates some sub-lists in reverse order, similar to the recursive construction of the BRGC.  The recursive properties allow us to provide efficient ranking and unranking algorithms for the listing based on counting the number of trees at each stage of the construction.   Let \spt{n} denote the number of spanning trees of $F_n$. It is known that 
\[\spt{n} = f_{2(n-1)} =  2 \frac{((3-\sqrt{5})/2)^{n}-((3+\sqrt{5})/2)^{n-2}}{5-3\sqrt{5}},\] 
where $f_n$ is the $n$th number of the Fibonacci sequence with $f_1=f_2=1$ \cite{fanformula}.

By studying the order of the spanning trees in \textsc{Greedy}$(P_n)$,
we identified four distinct stages S1, S2, S3, S4 that are highlighted for \textsc{Greedy}$(P_6)$ in 
Fig.~\ref{fig:F6Generation}.  From this figure, and referring back to Fig.~\ref{fig:F2F3F4F5} to see the recursive properties, observe that:
\begin{itemize}
    \item The trees in S1 are  equivalent to \textsc{Greedy}$(P_{5})$ with the added edge $v_6 v_{5}$. 
    
    \item The trees in S2 are equivalent to the reversal of the trees in \textsc{Greedy}$(P_{5})$ with the added edge $v_6 v_{\infty}$. 
\end{itemize}
\noindent
The trees in S3 and S4 have both edges $v_6 v_{5}$ and $v_6 v_{\infty}$ present.  
\begin{itemize}

    \item In S3,  focusing only on the vertices $v_4,v_3,v_2,v_\infty$, the induced subgraphs correspond to  \textsc{Greedy}$(P_{4})$, except whenever  $v_4v_\infty$ is present, it is replaced with $v_4v_5$ (the last five trees).
    \item  In S4, focusing only on the vertices $v_4,v_3,v_2,v_\infty$, the induced subgraphs correspond to the trees in \textsc{Greedy}$(P_{4})$
    where $v_4v_\infty$ is  present, in reverse order.
\end{itemize} 

Generalizing these observations for all $n \geq 2$ leads to the recursive procedure  {\sc Gen}($k,s_1,\mathit{varEdge}$) given in Algorithm~\ref{alg: Gen}, which uses a global variable $T$ to store the current spanning tree with $n$ vertices. The parameter $k$ indicates the number of vertices under consideration; the parameter $s_1$ indicates whether or not to generate the trees in stage S1, as required by the trees for S4; and the parameter $\mathit{varEdge}$ indicates whether or not a variable edge needs to be added as required by the trees for S3.  The procedure {\sc RevGen}($k,s_1,\mathit{varEdge}$), which is left out due to space constraints, simply performs the operations from {\sc Gen}($k,s_1,\mathit{varEdge}$) in reverse order.
For each algorithm the base cases correspond to the edge moves in the listings  \textsc{Greedy}$(P_2)$ and \textsc{Greedy}$(P_3)$.  

Let $\LIST{n}$ denote the listing obtained by initializing $T$ to $P_n$, printing $T$, and  calling {\sc Gen}($n,1,0$).  Let $L_n$ denote the last tree in this listing. Let $\REVLIST{n}$ denote the listing obtained by initializing $T$ to $L_n$, printing $T$, and  calling {\sc RevGen}($n,1,0$). Thus, $\REVLIST{n}$ is the  the listing $\LIST{n}$ in reverse order. 

\begin{algorithm}
    \footnotesize \caption{}
    \label{alg: Gen}
	\begin{algorithmic}[1]
	    \Procedure{Gen}{$k, s_1, varEdge$}
		\If{$k = 2$} \Comment{$F_2$ base case}
			\If{$varEdge$} $T \gets T - v_2 v_\infty + v_2 v_3$; \textsc{Print}$(T)$
			\EndIf
		\ElsIf{$k = 3$} \Comment{$F_3$ base case}
			 \If{$s_1$} 
				\If{$varEdge$} $T \gets T - v_3 v_2 + v_3 v_4$; \textsc{Print}$(T)$
				\Else{ $T \gets T - v_3 v_2 + v_3 v_\infty$}; \textsc{Print}$(T)$
				\EndIf
			\EndIf	
		 \State $T \gets T - v_2 v_\infty + v_2 v_3$; \textsc{Print}$(T)$
		\Else 
			\If{$s_1$} 
				\State \Gen{}$(k-1, 1, 0)$	 \Comment{S1}
				\If{$varEdge$} $T \gets T - v_k v_{k-1} + v_k v_{k+1}$; \textsc{Print}$(T)$
				\Else{ $T \gets T - v_k v_{k-1} + v_k v_\infty$}; \textsc{Print}$(T)$
				\EndIf
			\EndIf 
			\State \RevGen{}$(k-1, 1, 0)$ \Comment{S2}
			\State $T \gets T - v_{k-1} v_{k-2} + v_{k-1} v_k$; \textsc{Print}$(T)$
			\State	\Gen{}$(k-2, 1, 1)$  \Comment{S3}
			\If{$k > 4$} $T \gets T - v_{k-2} v_{k-1} + v_{k-2} v_\infty$; \textsc{Print}$(T)$
			\EndIf 
			\State \RevGen{}$(k-2, 0, 0)$  \Comment{S4}
		\EndIf
	    \EndProcedure
	\end{algorithmic}
\end{algorithm}

Our goal is to show that $\LIST{n}$ exhaustively lists all trees in $\trees{n}$ and moreover, the listing is equivalent to {\sc Greedy}($P_n$). We accomplish this in two steps: first we show that $\LIST{n}$ has the required size, then we show that $\LIST{n}$ is equivalent to {\sc Greedy}($P_n$).  Before  doing this, we first comment on some notation.
Let $T - v_i$ denote the tree obtained from $T$ by deleting the vertex $v_i$  along with all edges that have $v_i$ as an endpoint. Let $T + v_i v_j$ (resp. $T-v_i v_j$) denote the tree obtained from $T$ by adding (resp. deleting) the edge $v_i v_j$. For the remainder of this section, we will let $T_n$ denote the tree $T$ specified as a global variable for \Gen{} and \RevGen{}, and we let $T_{n-1}=T-v_n$ and $T_{n-2}=T-v_n-v_{n-1}$.  

\begin{lemma} \label{EdgeMovesLemma}
For $n \geq 2$, $|\LIST{n}| = |\REVLIST{n}| = \spt{n}$.
\end{lemma}

\begin{proof}
This result applies the Fibonacci recurrence and straightforward induction by counting the number of trees recursively generated in each stage S1, S2, S3, S4 as described earlier in this section. The base cases for $n=2,3,4$ are easily verified by stepping through the algorithms. A formal proof is omitted due to space constraints.~\hfill $\Box$
\end{proof}

To prove the next result, we first detail some required terminology.
If $T \in \trees{n}$, then we say that the operation of deleting an edge $v_i v_j$ and adding an edge $v_i v_k$ is a \emph{valid} edge move of $T$ if the result is a tree in $\trees{n}$ that has not been generated yet. Conversely, if the result is not a tree in $\trees{n}$, or the result is a tree that has already been generated, then it is not a \emph{valid} edge move of $T$. We say an edge $v_i v_j$ is \emph{smaller} than edge $v_i v_k$ if $j<k$. An edge move $T_n - v_i v_j + v_i v_k$ is said to be \emph{smaller} than another edge move $T_n - v_x v_y + v_x v_z$ if $i<x$, if $i=x$ and $j<y$, or if $i=x$, $j=y$, and $k<z$.

\begin{lemma} \label{Gen=GreedyTheorem}  
For $n\geq 2$,  $\LIST{n} = \textsc{Greedy}(P_n)$ and $\REVLIST{n} = \textsc{Greedy}(L_n)$.
\end{lemma}

\begin{proof}
By induction on $n$.  It is straightforward to verify that the result holds for $n=2,3,4$ by iterating through the algorithms. Assume $n>4$, and that $\LIST{j} = \textsc{Greedy}(P_j)$ and $\REVLIST{j} = \textsc{Greedy}(L_j)$ for $2\le j<n$. We begin by showing $\LIST{n} = \textsc{Greedy}(P_n)$, breaking the proof into each of the four stages for clarity. \\

\noindent \underline{S1:} Since $n>4$ and $s_1=1$, \Gen{}$(n-1, 1, 0)$ is executed. By our inductive hypothesis, $\LIST{n-1} = \textsc{Greedy}(P_{n-1})$. These must be the first trees for \textsc{Greedy}$(P_n)$, as any edge move involving $v_n v_{n-1}$ or $v_n v_\infty$ is larger than any edge move made by \textsc{Greedy}$(P_{n-1})$. Since \textsc{Greedy}$(P_{n-1})$ halts, it must be that no edge move of $T_{n-1}$ is possible. So \textsc{Greedy}$(P_n)$ must make the next smallest edge move, which is $T_n - v_n v_{n-1} + v_n v_\infty$. Since $T_n$ is a spanning tree, it follows that $T_n - v_n v_{n-1} + v_n v_\infty$ is also a spanning tree (and has not been generated yet), and therefore the edge move is valid. At this point, \Gen{}$(n, 1, 0)$ also makes this edge move, by line 13.\\
\begin{sloppypar}
\noindent \underline{S2:} \RevGen{}$(n-1, 1, 0)$ ($T_{n-1} = L_{n-1}$) is then executed. By our inductive hypothesis, $\REVLIST{n} = \textsc{Greedy}(L_{n-1})$. Since \textsc{Greedy}$(L_{n-1})$ halts, it must be that no edge moves of $T_{n-1}$ are possible. At this point, $T_{n-1} = P_{n-1}$ because \RevGen{}$(n-1, 1, 0)$ was executed. The smallest edge move now remaining is $T_n - v_{n-2} v_{n-1} + v_n v_{n-1}$. This results in $T_n = P_{n-2} + v_n v_{n-1} + v_n v_\infty$, which is a spanning tree that has not been generated. So, \textsc{Greedy}$(P_n)$ must make this move. \Gen{}$(n, 1, 0)$ also makes this move, by line 15. So, $\LIST{n}$ must equal \textsc{Greedy}$(P_n)$ up to the end of S2.\\

\noindent \underline{S3:} Next, \Gen{}$(n-2, 1, 1)$ starting with $T_{n-2} = P_{n-2}$ is executed. Since $varEdge = 1$, $v_{n-2} v_{n-1}$ is added instead of $v_{n-2} v_\infty$. \textsc{Greedy}$(P_n)$ also adds $v_{n-2} v_{n-1}$ instead of $v_{n-2} v_\infty$ since $v_{n-2} v_{n-1}$ is smaller than $v_{n-2} v_\infty$ and this edge move results in a tree not yet generated. Other than the difference in this one edge move, which occurs outside the scope of $T_{n-2}$, \Gen{}$(n-2, 1, 0)$ and \Gen{}$(n-2, 1, 1)$ (both starting with $T_{n-2}=P_{n-2}$) make the same edge moves. Since we also know that $\LIST{n-2} = \textsc{Greedy}(P_{n-2})$ by the inductive hypothesis, it follows that $\LIST{n}$ continues to equal \textsc{Greedy}$(P_n)$ after line 16 of \Gen{}$(n,1,0)$ is executed. We know that $T_{n-2} = L_{n-2}$ after \Gen{}$(n-2, 1, 0)$. However, $T_{n-2} = L_{n-2} - v_{n-2} v_\infty + v_{n-2} v_{n-1}$ instead because \Gen{}$(n-2, 1, 1)$ was executed ($varEdge=1$). It must be that no edge moves of $T_{n-2}$ are possible because \textsc{Greedy}$(P_{n-2})$ (and \Gen{}$(n-2, 1, 1)$) halted. The smallest edge move now remaining is $T_n - v_{n-2} v_{n-1} + v_{n-2} v_\infty$. This results in $T_{n-2} = L_{n-2}$. Also, $T_n = T_{n-2} + v_n v_{n-1} + v_n v_\infty$ is a spanning tree since $T_{n-2}$ is a spanning tree of $F_{n-2}$. So \textsc{Greedy}$(P_n)$ makes this move. \Gen{}$(n, 1, 0)$ also makes this move, by line 17, and thus $\LIST{n} = \textsc{Greedy}(P_n)$ up to the end of S3. \\
\end{sloppypar}
\noindent \underline{S4:} Finally, \RevGen{}$(n-2, 0, 0)$  starting with $T_{n-2} = L_{n-2}$ is executed. By our inductive hypothesis, $\REVLIST{n-2} = \textsc{Greedy}(L_{n-2})$. From the recursive definition of \RevGen{}, it is clear that \RevGen{}$(n-2, 0, 0)$ and \RevGen{}$(n-2, 1, 0)$ make the same edge moves until \RevGen{}$(n-2, 0, 0)$ finishes executing. So, by the inductive hypothesis, the listings produced by \RevGen{}$(n-2, 0, 0)$ and \textsc{Greedy}$(L_{n-2})$ are the same until this point, which is where \Gen{}$(n, 1, 0)$ finishes execution. By Lemma~\ref{EdgeMovesLemma} we have that $|\LIST{n}| = \spt{n}$. Therefore, \textsc{Greedy}$(P_n)$ has also produced this many trees, and each tree is unique. Thus, it must be that all $\spt{n}$ trees of $F_n$ have been generated. Thus, \textsc{Greedy}$(P_n)$ also halts.

Since $\LIST{n}$ and \textsc{Greedy}$(P_n)$ start with the same tree, produce the same trees in the same order, and halt at the same place, it follows that $\LIST{n} = \textsc{Greedy}(P_n)$. It is relatively straightforward to show that $\REVLIST{n} = \textsc{Greedy}(L_n)$ by using similar arguments as above. This proof is omitted due to space constraints. 
\hfill $\Box$
\end{proof}
Since $\LIST{n}$ is the reversal of $\REVLIST{n}$, we immediately obtain the following corollary.
\begin{corollary}
For $n \geq 2$, \textsc{Greedy}$(P_n)$ is equivalent to \textsc{Greedy}$(L_n)$ in reverse order.
\end{corollary}

Because \textsc{Greedy}$(P_n)$ generates unique spanning trees of $F_n$, Lemma~\ref{EdgeMovesLemma} together with Lemma~\ref{Gen=GreedyTheorem} implies our first main result.  This result answers {\bf Research Question \#3} and the first part of {\bf Research Question \#1} for fan graphs.
\begin{theorem}
For $n \geq 2$,  $\LIST{n}$ = \textsc{Greedy}$(P_n)$ is a pivot Gray code listing of $\trees{n}$.
\end{theorem}

To efficiently store the global tree $T$, the algorithms \Gen{} and \RevGen{} can employ an adjacency list model where each edge $uv$ is associated only with the smallest labeled vertex $u$ or $v$.  This means $v_\infty$ will never have any edges associated with it, and every other vertex will have at most 3 edges in its list.  Thus the tree $T$ requires at most $O(n)$ space to store, and edge additions and deletions can be done in constant time.  Our next result answers the second part of {\bf Research Question \#1} for fan graphs.

\begin{theorem}
For $n\geq 2$, $\LIST{n}$ and $\REVLIST{n}$ can be generated in $O(1)$-amortized time using $O(n)$ space.
\end{theorem}

\begin{proof}
For each call to \Gen{}$(n, s_1, varEdge)$ where $n>3$, there are at most four recursive function calls, and at least two new spanning trees generated.  Thus, the total number of recursive calls made is at most twice the number of spanning trees generated.  Each edge addition and deletion can be done in constant time as noted earlier. Thus each recursive call requires a constant amount of work, and hence the overall algorithm will run in $O(1)$-amortized time.  There is a constant amount of memory used at each recursive call and the recursive stack goes at most $n-3$ levels deep; this requires $O(n)$ space.  As mentioned earlier, the global variable $T$ stored as adjacency lists also requires $O(n)$ space.
\hfill $\Box$
\end{proof}

\subsection{Ranking and Unranking}

We now provide ranking and unranking algorithms for the listing $\LIST{n}$ of all spanning trees for the fan graph $F_n$. 

Given a tree $T$ in $\LIST{n}$, we calculate its rank by recursively determining which stage (recursive call) $T$ is generated. We can determine the stage by focusing on the presence/absence of the edges $v_n v_{n-1}$, $v_n v_\infty$, $v_{n-2} v_\infty$, and $v_{n-2} v_{n-1}$.  Based on the discussion of the recursive algorithm, there are $t_{n-1}$ trees generated in S1,  $t_{n-1}$ trees generated in S2,  $t_{n-2}$ trees generated in S3, and $t_{n-2} - t_{n-3}$ trees generated in S4. S3 is partitioned into two cases based on whether $v_{n-2} v_{n-1}$ ($varEdge)$ is present. For the remainder of this section we will let $T_{n-1} = T - v_n$ and $T_{n-2} = T - v_n - v_{n-1}$.

For $n>1$, let $R_n(T)$ denote the rank of $T$ in the listing $\LIST{n}$. If $n=2,3,4$, then $R_n(T)$ can easily be derived from Fig.~\ref{fig:F2F3F4F5}.  Based on the above discussion, for $n\geq 5$: 
\begin{small}
\begin{equation*} 
 R_n(T) = 
       \begin{cases}
		2 \spt{n-1} + 2 \spt{n-2} - R_{n-2}(T_{n-2}) + 1 
		& \text{if $e_1, e_2, e_3 \in T$} \\ 

		2 \spt{n-1} + R_{n-2}(T_{n-2} + e_3) 
		& \text{if $e_1, e_2, e_4 \in T$, $e_3 \not \in T$} \\

		2 \spt{n-1} + R_{n-2}(T_{n-2})
		& \text{if $e_1, e_2 \in T$, $e_3, e_4 \not \in T$} \\ 

		2 \spt{n-1} - R_{n-1}(T_{n-1}) + 1 & \text{if $e_2 \in T$, $e_1 \not \in T$} \\ 
		R_{n-1}(T_{n-1}) & \text{if $e_1 \in T$, $e_2 \not \in T$}
        \end{cases} 
\end{equation*}
\end{small}
where $e_1 = v_n v_{n-1}$, $e_2 = v_n v_\infty$, $e_3 = v_{n-2} v_\infty$, and $e_4 = v_{n-2} v_{n-1}$.

Determining the tree $T$ at rank $r$ in the listing $\LIST{n}$ follows similar ideas by constructing $T$ starting from a set of $n$ isolated vertices one edge at a time. Let $U_n(T, r, e)$ return the tree $T$ at rank $r$ for the listing $\LIST{n}$. Initially, $T$ is the set of $n$ isolated vertices, $r$ is the specified rank, and $e = v_n v_\infty$.
If $n=2,3,4$, then $T$ is easily derived from Fig.~\ref{fig:F2F3F4F5}. For these cases, if the edge $v_n v_\infty$ is present, then it is replaced by the edge $e$ that is passed in.
\begin{footnotesize}
\begin{equation*} 
 U_n(T, r, e) = 
       \begin{cases}
	U_{n-1}(T {+} e_1, r, v_{n-1} v_\infty) 
	& \text{if $0 < r \leq \spt{n-1}$,} \\ 

	U_{n-1}(T {+} e, 2 \spt{n-1} {-} r {+} 1, v_{n-1} v_\infty) 
	& \text{if $\spt{n-1} < r \leq 2 \spt{n-1}$,} \\ 

	U_{n-2}(T {+} e_1 {+} e, r {-} 2\spt{n-1}, e_4) 
	& \text{if $2 \spt{n-1} < r \leq 2 \spt{n-1} {+} \spt{n-2}$,} \\ 

	U_{n-2}(T {+} e_1 {+} e, 2 \spt{n-1} {+} 2 \spt{n-2} {-} r {+} 1, e_3)
	& \text{otherwise.}
        \end{cases}  
\end{equation*}
\end{footnotesize}
where $e_1 = v_n v_{n-1}$, $e_3 = v_{n-2} v_\infty$, and $e_4 = v_{n-2} v_{n-1}$.

Since the recursive formulae to perform the ranking and unranking operations each perform a constant number of operations and the recursion goes $O(n)$ levels deep, we arrive at the following result provided the first $2(n{-}2)$ Fibonacci numbers are precomputed. We note that the calculations are on 
numbers up to size $t_{n-1}$.

 \begin{theorem}
 The listing $\LIST{n}$  can be ranked and unranked in $O(n)$ time using $O(n)$ space under the unit cost RAM model.
 \end{theorem}
 
 This answers {\bf Research Question \#2} for fan graphs.


\section{Conclusion}

We answer each of the three Research Questions outlined in Section~\ref{sec:intro} for the
fan graph, $F_n$.  First, we discovered a greedy algorithm that exhaustively listed all spanning trees of $F_n$ experimentally for small $n$ with an easy to define starting tree.  
We then studied this listings which led to a recursive construction producing the same listing that runs in $O(1)$-amortized time using $O(n)$ space.  We also proved that the greedy algorithm does in fact exhaustively list all spanning trees of $F_n$ for all $n\geq 2$, by demonstrating the listing is equivalent to the aforementioned recursive algorithm.  It is the first greedy algorithm known to exhaustively list all spanning trees for a non-trivial class of graphs. Finally, we provided an $O(n)$ time ranking and unranking algorithms for our listings, assuming the unit cost RAM model. It remains an interesting open problem to answer the research questions for other classes of graphs including the wheel, $n$-cube, and complete graph.


\bibliographystyle{splncs04}
\bibliography{FinalConferencePaper}
\normalsize

\end{document}